\documentclass[conference]{IEEEtran} 
\usepackage[cmex10]{amsmath}
\usepackage{amssymb,amsbsy,amsfonts,bbm,mathrsfs,amsthm,amscd,array}
\usepackage{graphicx, textcomp, verbatim}

\usepackage{amsmath,amssymb,amsbsy,amsfonts,graphicx,textcomp,bbm}

\newtheorem{thm}{Theorem}[section]

\newtheorem{prepro}{{\bf Proposition}}[section]

\newtheorem{precor}{{\bf Corollary}}[section]

\newtheorem{preconj}{{\bf Conjecture}}

\newtheorem{preremark}{{\bf Remark}}[section]

\newtheorem{predef}{{\bf Definition}}

\newtheorem{prelem}{{\bf Lemma}}[section]
\newenvironment{lem}{\begin{prelem}{\hspace{-0.5
               em}{\bf.}}}{\end{prelem}}
               
\newtheorem{preclaim}{{\bf Claim}}[section]

\newtheorem{prefact}{{\bf Fact}}[section]

\def\ux {\underline{x}}
\def\uth{\underline{\theta}}
\def\uxa {\underline{x}_{\partial{a}}}

\def\uxai {\underline{x}_{\partial{a} \backslash i}}

\def\E{\mathbb{E}}

\def\arg{\text{\rm{arg}}}
\def \cX{{\mathcal{X}}}

\def\da{\partial a}
\def\di{\partial i}
\def\xda{\ux_{\partial a}}

\def\cost{{\cal C}}
\def\obj{{\cal O}}
\def\utheta{\underline{\theta}}
\def\ux{\underline{x}}
\def\ur{\underline{r}}

\def\reals{{\mathbb R}}
\def\MARG{{\rm MARG}}
\def\LOC{{\rm LOC}}

\def\cS{{\cal S}}

\def\bbI{{\mathbb{I}}}

\def\optproblemzero{{(\text{P}_0)}}
\def\optproblem{{(\text{P})}}
\def\argmin{{\arg \!\min}}

\def\maxprod{{\sc Max-Product \;}}
\def\robust{{\sc Robust Max-Product \;}}

\def\cvx{{\sc cvx}}

\begin{document}
\title{Robust Max-Product Belief Propagation}

\author{\IEEEauthorblockN{Morteza Ibrahimi${}^{*}$, Adel Javanmard${}^{*}$, Yashodhan Kanoria${}^{*}$ and Andrea Montanari${}^{* \dagger}$}
\IEEEauthorblockA{${}^{*}$Department of Electrical Engineering, Stanford University\\
${}^{\dagger}$Department of Statistics, Stanford University\\
Email: \{ibrahimi, adelj, ykanoria, montanar\}@stanford.edu}
}
\maketitle

\begin{abstract}
We study the problem of optimizing a graph-structured objective function
under \emph{adversarial} uncertainty. This problem can be modeled 
as a two-persons zero-sum game between an Engineer and Nature.
The Engineer controls a subset of the variables (nodes in the graph),
and tries to assign their values to maximize an objective function. 
Nature controls the complementary subset of variables and tries to
minimize the same objective. This setting encompasses 
estimation and optimization problems under model uncertainty, and
strategic problems with a graph structure.
Von Neumann's minimax theorem guarantees the existence of a (minimax)
pair of randomized strategies that provide optimal robustness for each
player against its adversary. 

We prove several structural properties
of this strategy pair in the case of graph-structured payoff 
function. In particular, the randomized minimax strategies (distributions
over variable assignments) can be chosen in such a way to satisfy
the Markov property with respect to the graph. This significantly
reduces the problem dimensionality. 
Finally we introduce a message passing algorithm to 
solve this minimax problem. The algorithm generalizes 
max-product belief propagation to this new domain.
\end{abstract}

\section{Introduction}
\label{sec:introduction}

A two-persons zero-sum game in normal form is specified by an objective
(or utility)
function $\obj: (x,\theta)\mapsto \obj(x,\theta)$, whereby $x\in\cX$ is the
strategy of the first player (which we shall call by convention the
Engineer), while $\theta\in \Theta$ is the strategy of the second player
(Nature). Once the strategy pair $(x,\theta)$ is chosen, the Engineer
earns from  Nature an amount $\obj(x,\theta)$. The two players
optimize their strategies with respect to the opposite objective of
maximizing (Engineer) or minimizing (Nature) the objective.
Zero-sum games capture strategic situations in which agents compete
for a fixed, limited pool of resources
\cite{VonNeumannBook,BinmoreBook,Nisan:2007:AGT:1296179}. 
Remarkably, they have found  broad applicability beyond 
economic theory, including areas such as online prediction and learning 
\cite{CesaBianchi}, and statistical decision theory \cite{BergerBook}.
Here statistical estimation is viewed as a game between a Statistician
(who tries to design the best statistical procedure) and Nature (who
chooses the worst parameters).

Closer to our motivation, 
a large variety of optimal design problems in engineering can
be reduced to maximizing an appropriate objective function. The form of
this function is normally dictated by a model of the underlying
system, with parameters to be estimated empirically. Of course the parameter
estimation process is inherently imprecise and, more importantly, any
model of a
real system necessarily overlooks a multitude of effects.
This remark has motivated the burgeoning   fields of \emph{robust
  optimization}  and \emph{robust control} \cite{NemirovskiBook}. 
In this context, one considers a family of objective functions
$x\mapsto \obj(x;\theta)$, with $x$ the design variables, and $\theta\in
\Theta$ a vector of parameters. Rather than designing for a `nominal' $\theta_*\in\Theta$,
one then tries to maximize the worst case cost
$\min_{\theta\in\theta}\obj(x;\theta)$. The problem is hence reduced
to a two-players zero-sum game.

Robust optimization theory provides a wealth of structural
information, and efficient algorithms 
for classes of objective functions $\obj(\,\cdot\,,\theta): x\mapsto \obj(x,\theta)$
that are convex in the control variables $x$. The present paper takes 
a complementary point of view. We assume that both $x$ and $\theta$
take values in  high-dimensional, discrete spaces. Explicitly,
$x=\ux\in \cX^{V}$ and $\theta=\utheta\in \Theta^{F}$
where $\cX$, $\Theta$ are finite alphabets and $V$, $F$ are finite
index sets. Letting $|V|=n$ and $|F|=m$, a pair of pure strategies is specified by two vectors:
the Engineer controls variables $\ux= (x_1,x_2,\dots,x_n)$ indexed by the elements of $V$, 
while Nature controls parameters $\utheta= (\theta_1,\theta_2,\dots,\theta_n)$.
Within this setting we aim at finding strategies for the Engineer
which are optimally robust with respect to Nature.

Of course, this general problem is NP-hard, and indeed so
even in absence of any adversary (since it includes MaxSat as a
special case). Our approach is to exploit 
simplifications that follow from the underlying factorization
structure of the objective function. More precisely, we shall assume that
the objective $\obj(\ux,\utheta)$ can be expressed as a sum of terms which
are local on a graph $G=(V,F,E)$, whereby $V$ is the set of nodes
controlled by the Engineer, $F$ the set of nodes controlled by Nature,
and $E$ the edge set. 

Graph-structured objective functions naturally arise from probabilistic
graphical models \cite{KollerBook}. In particular, if $\mu(\ux)$ is the probability of
configuration $\ux$ under a probabilistic graphical model, then
$\log \mu(\ux)$ is an objective function that factors
additively. Hence MAP estimation falls in the class of
optimization problems considered here. With a slight abuse of
terminology, we shall use the term `graphical model' to refer to
general graph-structured objective functions, even if these are not
originated from probability distributions.

The application to graphical models also clarifies the need for robustness. 
Graphical models are particularly effective at expressing complex relationships.
Think for instance to the subtle relationships between diseases and symptoms in a medical
diagnostic systems \cite{GraphMed}. 
Such relationships are normally modeled
through simple parametric families of conditional probabilities
(e.g. logit or noisy OR). However, it is not expected that these
parametric expressions coincide with the 
`true' conditional distributions. The only solid justification for
this methodology is that the resulting predictions are robust with
respect to the details of the model itself. Robustness is therefore
implicitly assumed, but has never been carefully
investigated and accounted for  (but see Section \ref{sec:related_work} for related work).

Apart from the use of graphical models, we achieve significant structural
simplification by convexifying the space of strategies, i.e. introducing
randomized (mixed) strategies. This is a well
established path within game theory. 
The Engineer has at her disposal a stochastic device generating
strategy $\ux\in \cX^V$ with probability $p(\ux)$, and plays it,
while Nature plays strategy $\utheta\in\Theta^F$ with
probability $q(\utheta)$.
The Engineer tries to maximize the expected utility
$\E_{(p,q)}\{\obj(\ux,\utheta)\}$, while Nature tries to minimize the same
quantity. A crucial consequence is that the
problems faced by the two players  become dual linear programs (LPs). In
particular, the celebrated Von Neumann's minimax theorem ensures the
existence of  a saddle point, i.e. a strategy pair $(p^*,q^*)$ such
that for any other strategies $p$, $q$
\small 
\begin{eqnarray}
\E_{(p,q^*)}\{\obj(\ux,\utheta)\}\le 
\E_{(p^*,q^*)}\{\obj(\ux,\utheta)\}\le 
\E_{(p^*,q)}\{\obj(\ux,\utheta)\} \,. \label{eq:SaddlePoint}
\end{eqnarray}
\normalsize

This is equivalent to requiring that $(p^*,q^*)$ forms a Nash equilibrium.
The saddle point condition implies in particular that the order of play does not matter:
$\max_p\min_q\E_{(p,q)}\{\obj(\ux,\utheta)\}=
\min_q\max_p\E_{(p,q)}\{\obj(\ux,\utheta)\}$.
In words, $p^*$ provides to the Engineer optimal robustness against
Nature's adversarial choice, and indeed the same as if this choice was
known in advance. Remarkably, the worst-case expected utility of strategy
$p^*$ is in general strictly larger than the  utility of any
pure strategy.

Notice that convexification of the strategies space  is achieved at the expense of an
exponential blow-up in dimensionality. While a pure strategy for the 
Engineer is a (discrete) vector of length $n$, a mixed strategy is a (probability)
vector of length $|\cX|^n$. Hence, by itself, convexification does not reduce the
problem complexity.

In the next section we will 
illustrate key ideas and questions on a simple example,
and then describe  our general formalism and contributions.
In Section \ref{sec:Algorithm} we  derive a message passing algorithm,
called \robust
to construct minimax strategies. Finally, we review related work in
Section \ref{sec:related_work}.

\section{An example and main contributions}
\label{sec:examples_main}
\begin{figure}[t]
\phantom{a}
\vspace{-0.5cm}
\includegraphics[width=.46\textwidth]{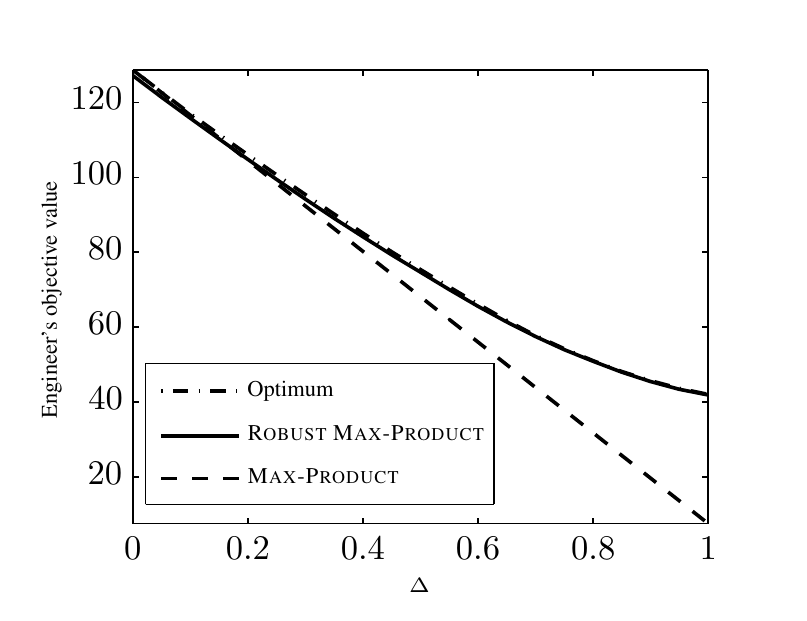}
\vspace{-0.5cm}
\caption{Engineer's objective vs. $\Delta$ for the Ising model example.}
\vspace{-0.25cm}
\label{fig:obj_vs_delta}
\end{figure}

\begin{figure}[t]
\phantom{a}
\vspace{-0.25cm}
\includegraphics[width=.43\textwidth]{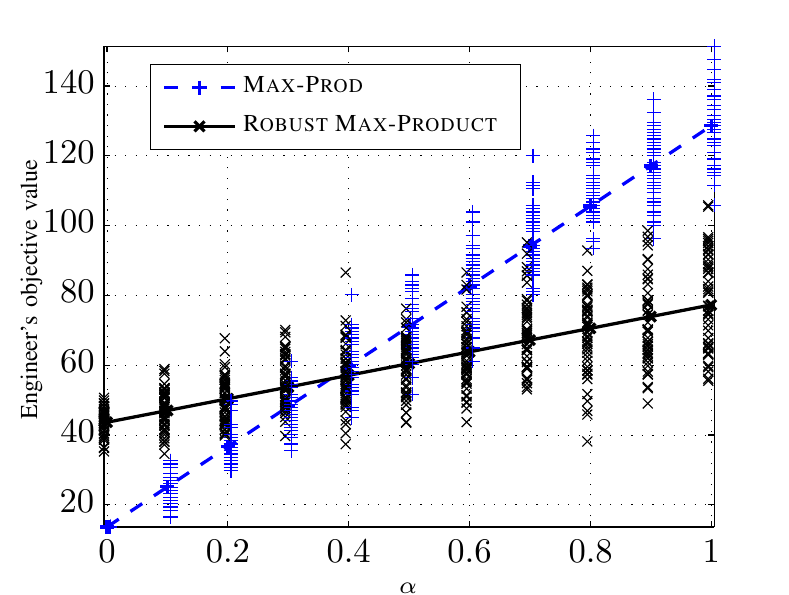} 
\vspace{-0.25cm}
\caption{Ising model example; the solid curves indicate the
  Engineer's expected payoff and the scattered points show the payoff for random
  instantiations of mixed strategies. For each value of $\alpha$, there are $50$ random instantiations of mixed strategies.}
\vspace{-0.25cm}
\label{fig:comp_robust_maxprod}
\end{figure}
\noindent {\bf Ising models.}
The Ising model is a pairwise graphical model with binary
alphabet $\cX = \{-1, +1\}$. The unnormalized log probability
$\obj(\ux, \uth) = \log\mu(\ux)+{\rm const.}$ can be written as 
\begin{eqnarray}\label{eq:ising_obj}
\begin{split}
\obj(\ux, \uth)= \sum_{(i,j)\in E}\psi_{ij}(x_ix_j;\theta_{ij}) + \sum_{i\in V}\psi_{i}(x_i;\theta_{i}),\\
\psi_{ij}(x_ix_j;\theta_{ij}) = \theta_{ij} x_i x_j, \;\psi_{i}(x_i;\theta_{i}) = \theta_{i} x_i.
\end{split}
\end{eqnarray}
In practice, the parameters $\uth$ are learned
from the data and hence are inaccurate. 
The Engineer's challenge is to find an strategy $p^*(\ux)$ which
maximize $\E_{(p(\ux), q^*(\uth))} \obj(\ux,\uth)$ for 
the worst case distribution $q^*$ of the uncertain parameter $\uth$.

Minimax strategies depend on the domain of $\utheta$.
We considered a family of models with parameters $h,\Delta\ge 0$. 
Single variable potentials have  $\theta_i \!\sim\! U[-h, h]$, i.e., $\theta_i$ is uniformly distributed between $-h$ and $h$, and
edges belong to two classes: \emph{positive} and 
\emph{negative}. Positive edges have $\theta_{ij}\in
\{+1-\Delta,+1,+1+\Delta\}$ while negative edges have $\theta_{ij}\in
\{-1-\Delta,-1,-1+\Delta\}$. 

Finding minimax strategies for this model is NP-hard,
even in the special case of $\Delta=0$, $h=0$. Indeed, when all edges are
negative, this reduces to the MAXCUT problem. 

We consider a random tree with $n = 93$ nodes as the underlying graph and perform the following experiments. 

\smallskip

\noindent \emph{First Experiment: } We apply the \robust algorithm 
derived in Section \ref{sec:Algorithm} to find the Engineer's optimum
strategy, $p^*$. 
We further compare it with the case the Engineer ignores Nature and apply
the classical \maxprod to the graph with nominal values;
namely $\theta_{ij} = +1$ on positive edges and $-1$ on negative edges.  Finally, in order to check the convergence of \robust, we compute $p^*$ by solving the minimax optimization problem in \cvx \cite{cvx}. Note that the latter is computationally much more expensive than the \robust and \maxprod. Figures~\ref{fig:obj_vs_delta} summarizes the results for different values of $\Delta$. 
\robust was run for $100$ iterations.
For $\Delta = 0$, \maxprod and \robust  are equivalent as Nature has no power. 
As $\Delta$ increases, \robust performs increasingly better compared to \maxprod. 

\smallskip

\noindent \emph{Second Experiment: } Let $p^*$ be the Engineer's strategy given by \robust algorithm, and let $\tilde{p}^*$ be the one obtained by applying \maxprod considering the nominal values of the parameters. Finally, denote by $q^*$ and $\tilde{q}^*$, Nature's best responses against $p^*$ and $\tilde{p}^*$. In this experiment, we compare the performance of strategies $p^*$ and $\tilde{p}^*$, when Nature deviates from her optimal strategies $q^*, \tilde{q}^*$. More specifically, Nature's strategy is chosen to be a mixture of her optimum strategy and the uniform distribution, i.e., $q(\uth)
= (1-\alpha) q^*(\uth) + \alpha |\Theta|^{-m}$ and $\tilde{q}(\uth) =
(1-\alpha) \tilde{q}^*(\uth) + \alpha |\Theta|^{-m}$.  Figure~\ref{fig:comp_robust_maxprod} illustrates the Engineer's payoff for the pairs of strategies $(p^*,q(\uth))$ and $(\tilde{p}^*,\tilde{q}(\uth))$, as $\alpha$ varies. Here, $\alpha =
0$ corresponds to the case Nature chooses her optimum
strategy. \robust outperforms the simple \maxprod in this regime. As
$\alpha$ increases, Nature changes from being adversarial to
completely random for $\alpha = 1$. \maxprod outperforms \robust 
in the latter case, since Nature is no-longer an adversary and one can
design for nominal values.    

\subsection{Main contributions}

We consider a general bipartite
graph (or  factor graph) $G =(V,F,E)$, where nodes in $V$  (\emph{variable
  nodes}, to be denoted by $i,j,k,\dots$) are controlled by the Engineer,  nodes in $F$  (\emph{factor
  nodes}, denoted by $a,b,c,\dots$) are controlled by Nature,  and $E \subseteq V \times F$ is a
set of undirected edges. 
Given $i \in V$ , the set of its neighbors is denoted by $\partial{i}
= \{a \in F : (i, a) \in E \}$. The neighborhood of $a \in F$ ,
denoted by $\partial{a}$, is defined analogously. 

The objective function $\obj:\cX^V\times\Theta^F\to\reals$
 \emph{factors} on graph $G$ if\footnote{While slightly more
   general definitions (symmetric in $\ux$ and $\utheta$) are
   possible, we stick to the present one
   because it is already rich enough to discuss all the key challenges.}  there exists a set of functions
 $\underline{\psi} = \{\psi_a : a \in F\}$, $\psi_a : \cX ^{\partial{a}} \times\Theta\to \reals$ such that
\begin{equation*}
\obj(\ux,\uth) = \sum_{a \in F} \psi_a(\uxa;\theta_a).
\end{equation*}
The functions $\psi_a$ are called \emph{potentials}. There is no loss
of generality in assuming $G$ to be bipartite. If two nodes $i$, $j$
controlled by the same player were neighbors, we could replace them by
a single node with strategy space $\cX'\equiv\cX\times\cX$.

As discussed above, our goal is to find a pair $(p^*,q^*)$ where 
$p^*$ is a probability distribution over $\cX^{V}$, $q^*$ a
distribution over $\Theta^F$, and the pair satisfies the Nash
equilibrium condition  (\ref{eq:SaddlePoint}).
From the point of view of the Engineer, this amounts to solving the problem
\begin{equation}
\label{eqn:main_problem}
p^*= \arg \max_{p} \min_{q} \E_{(p, q)} \obj(\ux,\uth)\, .
\end{equation}
The \emph{support}  $\text{supp}(p)$  of a probability distribution
$p$ is the smallest set $S$ such that $p(S^c) = 0$. 

Since $\E_{(p, q)} \obj(\ux,\uth)$ is linear both in $p$ and $q$,
which belong to the simplex, (\ref{eqn:main_problem})  is equivalent
to an LP problem.
However, the dimensionality of this problem is exponential in the
graph size: even writing down the strategy takes exponential time.
The following result plays a key role in our
approach. 
\begin{thm}
\label{rem:MRF_prod}
In problem~\eqref{eqn:main_problem}, without loss of generality, we
can assume that the Engineer's strategy is a Markov Random Field (MRF)
with factor graph $G$, and that Nature chooses a product distribution.
Explicitly, the Engineer's strategy can be assumed to take the form 
$p^*(\ux) = \prod_{a\in F} f_a(\uxa)$, while Nature's strategy takes
the form  $q^*(\utheta) = \prod_{a\in F}g_a(\theta_a)$.
\end{thm}
\begin{IEEEproof} 
Note that the Engineer's pay off is given by
{\small \begin{align*}
\E_{(p(\ux), q(\uth))} \sum_{a \in F} \psi_a(\uxa;\theta_a)
= \sum_{a \in F} \E_{q_a (\theta_a)} \E_{p_a(\uxa)} \psi_a(\uxa; \theta_a).
\end{align*}}

Therefore, only the marginals $q_a (\theta_a)$ of the Nature's  distribution play role in the pay off. Hence, we can assume that Nature has a product distribution $q(\uth) = \prod_{a \in F} q_a (\theta_a)$.

Similarly, only the marginals $p_a(\uxa)$ appears in the pay off. Thereby, without loss of generality we can assume that the Engineer's distribution is an MRF with respect to $G$. This follows from the fact that  for all factor graphs $G$ and for all joint distributions $p(\ux)$, there exists a distribution $\tilde{p}(\ux)$ that is representable as an MRF with graph $G$ such that $p_a(\uxa) = \tilde{p_a}(\uxa)$~\cite{Nisan:2007:AGT:1296179}. 
\end{IEEEproof}

Notice that, for a graph $G$ with bounded degree, a MRF can be
specified by $O(|V|)$ parameters. In particular, the MRF is completely
specified by the marginals $p_a(\ux_{\da})$. We  then reformulate
the minimax problem as the one of computing the minimax
marginals $\{p_a^*\}_{a\in F}$. By definition, these belong to
the so-called marginal polytope \cite{WainwrightJordan}
\begin{align*}
\MARG(G) \equiv \Big\{\{p_a\}_{a\in F} & \Big|\;
p_a(\uxa) = \sum_{\ux_{V\backslash \partial{a}}} p(\ux),\\
     & \text{ for some distribution } p(\ux)
\Big\}\, .
\end{align*}

\noindent Problem \ref{eqn:main_problem} can therefore be restated as an LP over
$\MARG(G)$:
\begin{equation}
\begin{split}
\label{eqn:minimax}
&\underset{p_a(\uxa)}{\text{maximize}} \quad \sum_{a \in F}
\underset{\theta_a}{\min} \Big[ \sum_{\underline{x}_{\partial a}}
p_a(\uxa)\psi_a(\underline{x}_{\partial a}; \theta_a) \Big],\\
&\text{s.t.} \quad \quad \quad \quad \{p_a\}_{a\in F} \in \MARG(G)\, .
\end{split}
\end{equation}
Here we used the fact that the $\min$ over $q$ in the simplex is
necessarily achieved at one extremal point, i.e. at a pure strategy.
In general, $\MARG(G)$ does not possess a polynomial separation oracle
and therefore this problem is not tractable. Instead, we relax it to the
set of locally consistent marginals on $G$, denoted by $\LOC(G)$  
\begin{equation}
\begin{split}
\label{eqn:LOC_def}
&\LOC(G)=\\
&\Bigg\{ \{p_a\}_{a\in F}\;
\begin{array}{|ll}
\exists p_i(x_i) : p_i(x_i) \geq 0, & \underset{x_i}{\sum} p_i(x_i) = 1\\
p_a(\uxa) \geq 0, & \underset{\ux_{\partial{a} \backslash i}}{\sum} p_a(\uxa) = p_i(x_i) 
\end{array}
\Bigg\}.
\end{split}
\end{equation}
We then have the following relaxation of problem~\eqref{eqn:minimax} to the local polytope. 
\hspace{-1cm}
\begin{equation}
\begin{split}
\label{eqn:minimax_r}
 (\text{P}_0): \quad &\underset{p_a(\uxa)}{\text{maximize}} \quad \sum_{a \in F}
\underset{\theta_a}{\min} \Big[ \sum_{\underline{x}_{\partial a}}
p_a(\uxa)\psi_a(\underline{x}_{\partial a}; \theta_a) \Big]\, ,\\
&\text{s.t.} \quad \quad \quad \quad \{p_a\}_{a\in F} \in \LOC(G)\, .
\end{split}
\end{equation}
If $G$ is  a tree, then $\text{\rm{LOC}}(G) \equiv
\text{\rm{MARG}}(G)$ and therefore this relaxation is exact. 

\section{Algorithm}

\label{sec:Algorithm}
\vspace{-0.1cm}

Here, we first present the alternating direction method of multipliers (ADMM) \cite{gabay1976dual}, \cite{glowinski1975approximation} algorithm for solving convex optimization problems and state a general result regarding its convergence properties.
Subsequently, we show how the optimization problem $\optproblemzero$ can be transformed to conform with the general form for the ADMM algorithm.
We derive the \robust algorithm from the transformed variant of the problem $\optproblemzero$ and obtain convergence guarantees using the result stated for the general case of ADMM algorithms.

\subsection{ADMM Algorithm}
What follows is a short presentation of the ADMM algorithm and its properties. 
The reader interested in a more comprehensive treatment can refer to \cite{boyd2010distributed}.
Consider the optimization problem 
\begin{equation}\label{eq:admm_optimization_form}
\begin{split}
&\underset{x \in \reals^n, z \in \reals^m}{\text{minimize}} \quad  f(x) + g(z),\\
&\text{s.t.} \quad \quad \quad \quad Ax-z=0.
\end{split}
\end{equation}
where $A\in \reals^{p\times n}$.
The \textit{augmented Lagrangian} for this problem is defined as
\begin{equation}
L_{\rho}(x, z, y) = f(x) + g(z) + y^T(Ax-z) +\frac{1}{2} \rho \|Ax-z\|_2^2.
\end{equation}
with $\rho >0$ a parameter and $\|\cdot\|_2$ indicating the $\ell_2$ norm.
The ADMM algorithm tries to solve the above optimization problem by starting from some initial estimates $(z^{(0)}, y^{(0)}= 0)$ and performing the following iteration
\begin{eqnarray}\label{eq:admm_update}
\begin{split}
x^{(t+1)} &= \underset{x}{\argmin} \; L_{\rho}(x, z^{(t)}, y^{(t)}) \, ,\\
z^{(t+1)} &= \underset{z}{\argmin} \; L_{\rho}(x^{(t+1)}, z, y^{(t)})\, ,\\ 
y^{(t+1)} &= y^{(t)} + \rho (A x^{(t+1)} - z^{(t+1)}).
\end{split}
\end{eqnarray}
The update rules in \eqref{eq:admm_update} closely resemble the dual gradient descent method where the dual is obtained from the augmented Lagrangian. This indeed is the gist of the \textit{method of multipliers}.
Despite the fact that the primal optimization is done in two steps and the augmented Lagrangian is used in place of the Lagrangian, the iteration \eqref{eq:admm_update} provably converges to the solution of \ref{eq:admm_optimization_form}.
Formally, assume the optimization problem \eqref{eq:admm_optimization_form} has a finite optimum value $p^*$. 
We say the Lagrangian $L(x,z, y)$ has a saddle point $(x^*,z^*, y^*)$ if $L(x^*,z^*, y) \le L(x^*,z^*, y^*) \le L(x,z, y^*)$ for all $x$, $z$, and $y$. 
Then the following theorem holds.
\begin{thm}\label{th:admm_convergence}
(\cite{gabay1976dual} Theorem 3.1, \cite{eckstein1992douglas} Theorem 8, \cite{boyd2010distributed} Section 3.2)
Assume that the extended real valued functions $f(x)$ and $g(z)$ are closed, proper, and convex and the un-augmented Lagrangian $L_0(x,z,y)$ has a saddle point. 
Then
\begin{align}
&\lim_{t \to \infty} f(x^{(t)}) + g(z^{(t)}) \to p^* \nonumber\\
&\lim_{t \to \infty} Ax^{(t)} - z^{(t)} \to 0 \nonumber\\ 
&\lim_{t \to \infty} y^{(t)} \to y^*
\end{align}
%
\end{thm}

\subsection{\robust Algorithm}

\begin{table*}[!t] \label{tbl:algorithm}
\caption{\robust Algorithm}
{\normalsize
\vspace{-10pt}
\noindent\rule{\textwidth}{1pt}
\vspace{-2pt}\\ 
{\bf Robust Max-Product:}\\
\vspace{4pt}
\rule{\textwidth}{1pt}\\
\vspace{1pt}
\textbf{Input:} Factor graph $G(V, F, E)$, potential functions $\{\psi_a(\xda, \theta_a\}_{a \in F}$\\
\textbf{Output:} Local marginals $\{p_a(\xda)\}_{a \in F}$\\
${\bf 1.}$ Initialize:
\begin{align*}
&\qquad  u^{(0)}_{ai}(x) = 0, \quad \forall \; (a,i) \in E, \; x \in \cX  \hspace{8cm}\\
&\qquad  p^{(0)}_{i}(x) = \frac{1}{|\cX|}, \quad \forall \; i \in V, \; x \in \cX
\end{align*}
${\bf 2.}$ Update until convergence: \\
\quad At the factor nodes:
\begin{equation*}
\begin{array}{l l l}
\qquad   \left\{p_a^{(t+1)}, \lambda_a^{(t+1)} \right\}  =  &\underset{p_a, \lambda_a}{\argmin} \; & \lambda_a + \frac{\rho}{2} \underset{\substack{i \in \da \\ x_i \in \cX}}{\sum} \left(\underset{\uxai}{\sum} p_a(\uxai; x_i) - p_i^{(t)}(x_i) - u_{ai}^{(t)}(x_i) \right)^2\\
\vspace{10pt}
& {\rm s.t.} & \lambda_a + \sum_{\uxa} p_a(\uxa)\psi_a(\uxa; \theta_a) \ge 0, \quad \forall{a,\theta_a} \\
\vspace{10pt}
& & p_a(\uxa) \ge 0, \hspace{3.8cm} \forall \xda \in \cX^{|\da|} \\
\end{array} \hspace{4cm}
\end{equation*}
\quad At the variable nodes: \\
\begin{align*}
&\qquad p_i^{(t+1)} = \Pi_{\cS^{|\cX| -1}} \left( \frac{1}{|\di|}\underset{a \in \di}{\sum}\left\{\underset{\uxai}{\sum} p^{(t)}_a(\uxai; \cdot)  - u_{ai}^{(t)}(\cdot) \right\}\right),\\
&\qquad u_{ai}^{(t+1)}(x_i) = u_{ai}^{(t)}(x_i) + \underset{\uxai}{\sum} p^{(t+1)}_a(\uxai; x_i) - p_i^{(t+1)}(x_i) \qquad \forall \; i, \; a \in \di, \; x_i \in \cX. \hspace{4cm}
\end{align*}
${\bf 3.}$ Return: $\{p_a\}_{a\in F}$. \vspace{5pt}\\
\rule{\textwidth}{1pt}
}
\end{table*}

Note that the epigraph form of the optimization problem $\optproblemzero$ is given by 
\begin{equation}\label{eq:optimization_epi}
\begin{array}{l l r}
{\text{minimize}}   &  \underset{a \in F}{\sum} \lambda_a &\\
\text{s.t.}  & \lambda_a + \underset{\uxa}{\sum} p_a(\uxa)\psi_a(\uxa; \theta_a) \ge 0, & \forall{a,\theta_a} \\
& p_i(x_i) = \underset{\uxai}{\sum} p_a(\uxa) , &\hspace{-12pt}\forall{(i,a) \in E,x_i} \\
&\sum_{x_i} p_i(x_i) = 1, &\forall{i \in V} \\
&p_a(\uxa) \geq 0,& \hspace{-17pt}\forall{(i,a) \in E,\uxa}\\
&p_i(x_i) \geq 0,  &\forall{i \in V,x_i}\\
&&
\end{array}
\end{equation}
where the minimization is over $\{\lambda_a\}_{a \in F}, \{p_a\}_{a \in F}, \{p_i\}_{i \in V}$ and $\LOC(G)$ is represented in terms of the set of marginals.
Define the indicator function $\mathbb{I}(\cdot)$ as
\begin{equation}
\mathbb{I}(x) = \left\{
\begin{array}{l l}
0 \qquad & \text{if } x = \text{TRUE},\\
\infty \qquad & \text{if } x = \text{FALSE}.
\end{array}
\right.
\end{equation}
%
Furthermore, let $f(\{\lambda_a\}_{a\in F}, \{p_a\}_{a\in F}) = \sum_{a \in F} \tilde{f}(\lambda_a, p_a)$ and $g(\{p_i\}_{i \in V}) = \sum_{i \in V} \tilde{g}(p_i)$ whereby $\tilde{f}(\lambda_a, p_a)$ and $\tilde{g}(p_i)$ are defined as
\begin{align}
\tilde{f}(\lambda_a, p_a)  &=  \lambda_a
 + \sum_{\theta_a\in \Theta }\bbI(\lambda_a + \sum_{\uxa} p_a(\uxa)\psi_a(\uxa; \theta_a) \ge 0) \nonumber\\
 &\;\; + \sum_{\xda \in \cX^{|\da|}} \bbI(p_a(\uxa) \geq 0), 
\end{align}
and 
\begin{align}
& \tilde{g}\left(p_i \right) =  \bbI(p_i \in \cS^{|\cX|-1}).
\end{align}
Here, $\cS^{|\cX|-1}$ is the $|\cX|-1$ dimensional simplex and $p_a$ and $p_i$ are the $|\cX|^{|\da|}$ and $|\cX|$ dimensional real vectors respectively.
It is easy to see that the extended real valued functions 
\begin{align*}
f&: \reals^{|F| + \sum_{a \in F}|\cX|^{|\da|}} \to (-\infty, +\infty] \\
g&: \reals^{|V| |\cX|} \to (-\infty, +\infty]
\end{align*}
are closed, convex, and proper.

Using the functions $f$ and $g$, the optimization problem \eqref{eq:optimization_epi} can be restated as
\begin{align*}
&\optproblem: \\
&\begin{array}{l l l }
{\text{minimize}}  & f(\{\lambda_a\}_{a\in F}, \{p_a\}_{a\in F}) + g(\{p_i\}_{i \in V}) & \\
\text{s.t.}              & p_i(x_i) = \underset{\uxai}{\sum} p_a(\uxai;x_i) ,                                    &\hspace{-10pt}\forall (a,i) \in E,\\
&& \hspace{-10pt} \forall x_i \in \cX,
\end{array}
\end{align*}
where the minimization is over $\{\lambda_a\}_{a \in F}$, $\{p_a\}_{a \in F}$, and $\{p_i\}_{i \in V}$. Optimization problem $\optproblem$ follows the form of the general problem in Eq. \eqref{eq:admm_optimization_form} and can be solved using the ADMM algorithm.
The augmented Lagrangian for problem $\optproblem$ can be written as
\begin{eqnarray}\label{eq:augmented_lagrangian}
\begin{split}
& L_{\rho}\left(\{\lambda_a\}_{a\in F}, \{p_a\}_{a\in F}, \{p_i\}_{i\in V}, \{u_{ai}\}_{(a,i)\in E}\right)=\\
&\quad \qquad \sum_{a \in F} \tilde{f}(\lambda_a, p_a) + \sum_{i \in V} \tilde{g} \left(p_i \right)  \\
& \quad  +\sum_{(a,i) \in E, x_i \in \cX}  u_{ai}(x_i) \, \Big(p_i(x_i) - \sum_{\uxai} p_a(\uxai;x_i) \Big)  \\
& \quad  +\sum_{(a,i) \in E, x_i \in \cX} \rho \, \Big(p_i(x_i) - \sum_{\uxai} p_a(\uxai;x_i)\Big)^2, 
\end{split}
\end{eqnarray}
where $\rho$ is a parameter and $u_{a,i}(x_i)$ are the dual variables.
Notice that the special form of the constraint results in the quadratic penalty being block separable in $\{p_a\}_{a \in F}$.
Furthermore, the function $f$ is also block separable in $\{p_a\}_{a \in F}$, as well. 
Similarly, the quadratic penalty and the function $g$ are separable in $\{p_i\}_{i \in V}$.
These facts enable us to further decompose the first two steps of the ADMM iteration (Eq. \eqref{eq:admm_update}) and perform the optimization at the corresponding check and variable node locally. The resulting algorithm is presented in Table~\ref{tbl:algorithm}.

\subsection{Convergence of the \robust Algorithm}
\begin{figure}[t]
\phantom{a}
\vspace{-0.25cm}
\includegraphics[width=.44\textwidth]{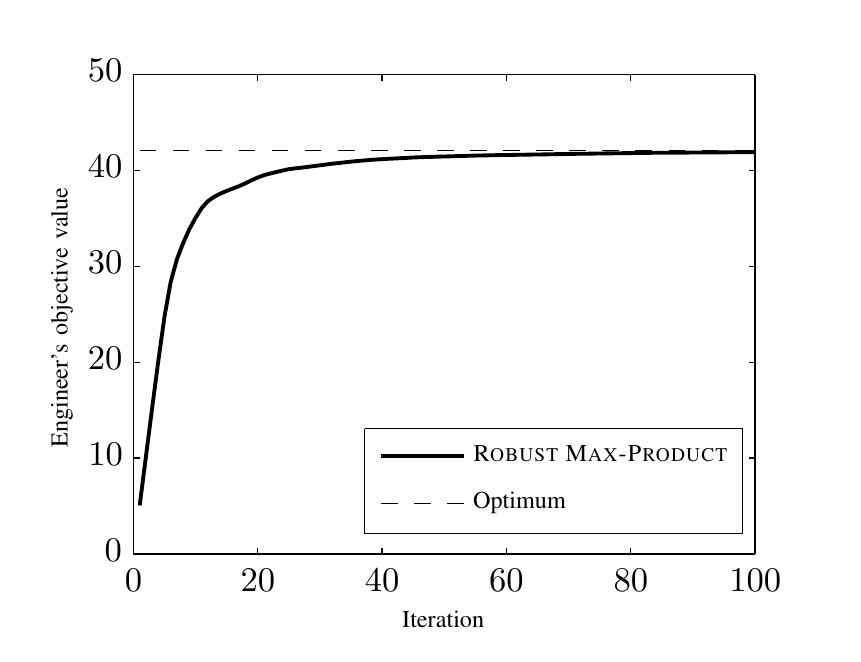}
\vspace{-0.25cm}
\caption{Engineer's objective vs. iteration.}
\vspace{-0.25cm}
\label{fig:obj_vs_itr}
\end{figure}
\begin{figure}[t]
\phantom{a}
\vspace{-0.25cm}
\includegraphics[width=.44\textwidth]{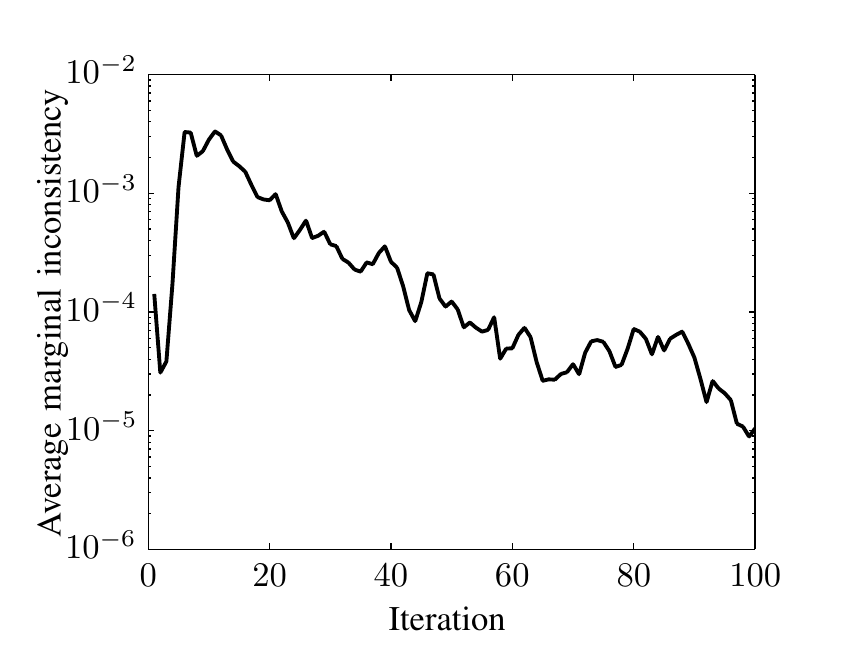}
\vspace{-0.25cm}
\caption{Average marginal inconsistency vs. iteration.}
\vspace{-0.25cm}
\label{fig:res_vs_itr}
\end{figure}

For $(a,i) \in E$, $x_i \in \cX$ and local marginals $\{p_a\}_{a \in F}$, $\{p_i\}_{i \in V}$,  define the marginal inconsistency residual $r_{ai}(x_i)$  as 
\begin{equation}
r_{ai}(x_i) = p_i(x_i) - \sum_{\uxai}{p_a(\uxai;x_i)}.
\end{equation}
Let $\ur \in \mathbb{R}^{|E|\times|\cX|}$ be the vector of residuals defined as $\ur = (\{r_{ai}(x_i)\}_{(a,i) \in E, x_i \in \cX})$. In particular, let $\ur^{(t)}$ be the marginal inconsistency residual at iteration $t$ of the \robust algorithm.

Define $\cost^{(t)} = \sum_{a \in F} \lambda_a^{(t)}$, the cost function at iteration $t$ of the \robust algorithm and let $\cost^*$ be the optimum value of the optimization problem \eqref{eqn:minimax_r}.
Figures~\ref{fig:obj_vs_itr} and \ref{fig:res_vs_itr} show $\cost^{(t)}$ and $\frac{1}{|E||\cX|} \|\ur^{(t)}\|_1 $ as a function of $t$ (iteration) for the Ising model described in Section \ref{sec:examples_main} with $\Delta = 1$. 
Figure~\ref{fig:res_vs_itr} shows that the marginals quickly become nearly consistent while Fig~\ref{fig:obj_vs_itr} demonstrates that the objective value converges to the optimum value as number of iteration increases to a modest number. 

The following theorem provides theoretical guarantees for the behavior observed in Figures \ref{fig:obj_vs_itr} and  \ref{fig:res_vs_itr}. In particular, it states that the sequence of local marginals in the \robust algorithm converges to a set of locally consistent marginals that achieves the optimum payoff for the Engineer.
\begin{thm}\label{th:convergence}
For any graph $G(F,V,E)$ and set of potential functions $\{\psi_a(\xda)\}_{a \in F}$ the followings hold.
\begin{itemize}
\item[{\rm (i)}] $\underset{t \to \infty}{\lim} \;\ur^{(t)} = \underline{0}.$
\item[{\rm (ii)}] $\underset{t \to \infty}{\lim} \; \cost^{(t)}  = \cost^*.$
\end{itemize}
\end{thm}
The proof of this theorem can be obtained by applying the result of the following lemma to Theorem~\ref{th:admm_convergence}.
\begin{lem}
For any graph $G(F,V,E)$ and potential functions $\{\psi_a(\xda)\}_{a \in F}$, the optimization problem \eqref{eqn:minimax_r} is feasible. 
Furthermore, there exist $\{\lambda_a^*\}_{a \in F}$, $\{p_a^*\}_{a \in F}$, $\{p_i^*\}_{i \in V}$, and $\{u_{ai}\}_{(a,i) \in E}^*$ such that $(\{\lambda_a^*\}_{a \in F}$, $\{p_a^*\}_{a \in F}$, $\{p_i^*\}_{i \in V}$, $\{u_{ai}\}_{(a,i) \in E}^*)$ is a saddle point of the augmented Lagrangian \eqref{eq:augmented_lagrangian} with $\rho = 0$ and $\sum_{a \in F} \lambda_a^* = \cost^*$.
\end{lem}
\begin{proof}
Consider the optimization problem in Eq. \eqref{eq:optimization_epi}. 
First note that given the potential functions $\{\psi_a(\xda)\}_{a \in F}$ are bounded, this problem has a strictly feasible point. Also it is a linear program. Hence the strong duality holds by Slater's theorem~\cite{Boyd} and the Lagrangian has a saddle point. Let $\{\lambda_a^*\}_{a\in F}, \{p_a^*\}_{a\in F}, \{p_i^*\}_{i\in V}$ be the values of the primal variables at this saddle point. 
Similarly, let $\{u_{ai}^*\}_{(a,i) \in E}$ be the values of the dual variables corresponding to the constraints $p_i(x_i) = \sum_{\uxai} p_a(\uxa)$ at this saddle point. 
Then $\{\lambda_a^*\}_{a\in F}, \{p_a^*\}_{a\in F}, \{p_i^*\}_{i\in V}$ are primal optimal for \eqref{eq:optimization_epi}.   
In particular,  $\sum_{a \in F} \lambda_a^* = \cost^*$.
Furthermore, it is easy to see that$(\{\lambda_a^*\}_{a\in F}, \{p_a^*\}_{a\in F}, \{p_i^*\}_{i\in V}, \{u_{ai}^*\}_{(a,i) \in E})$ is a saddle point of the Lagrangian in Eq. \eqref{eq:augmented_lagrangian} with $\rho=0$.

\end{proof}
\vspace{-0.1cm}

\section{Related work}
\label{sec:related_work}
\vspace{-0.1cm}

Several groups investigated the impact of graphical model structure on 
the computational properties of Nash equilibria 
\cite{Kearns,OrtizKearns,DaskalakisGraph,Elkind}. In particular,
Ortiz and Kearns \cite{OrtizKearns} proposed a message passing
algorithm (called {\sc NashProp}) to find 
Nash equilibria. However, as shown in \cite{Elkind}, the problem of
computing Nash equilibria is PPAD-complete  even on trees. 
Within graphical games studied in this literature, a
different player 
controls each vertex of a graph, and a game
is  played along each edge. A single player has at her disposal only a
small number of pure strategies (typically two), and the problem complexity
arises because of the large number of players.

\emph{Let us emphasize that the present paper studies a very different class
of models.} 
We consider a small fixed number of players (indeed 
in this paper only two players),  but each of them has at her disposal a
large number of pure strategies. The problem complexity is due to the
strategies proliferation.

The motivation for focusing on two-players zero-sum games
came from their relevance to
optimization and inference under model uncertainty.
A few authors \cite{Ihler,Varshney} have already analyzed the sensitivity 
of message passing algorithms to model uncertainty. However
these studies assumed a probability distribution over model
parameters, which is very restrictive in a high-dimensional setting,
or carried out a perturbation analysis, without constructing more
robust algorithms.

ADMM and many related algorithms (Uzawa's algorithm, Douglas-Rachford splitting, proximal method, Bregman iterative methods, etc.) have been around for a few decades. 
However, recent years have seen a surge of interest in these algorithms in many fields. 
The reader can refer to \cite{boyd2010distributed} for many examples in the field of statistical learning.
Closer to the spirit of this paper, \cite{ravikumar2010message} uses the technique of Bregman projection to obtain fractional solution for the maximum a posteriori probability (MAP) problem in graphical models. The problem addressed in this paper is fundamentally different from this work in that we consider the case of adversarial uncertainty in the model. 

\newpage
\bibliographystyle{abbrv}
\bibliography{references_aj2}

\newpage

\end{document}